\documentclass[12pt]{article}    
\usepackage[margin=0.75in]{geometry} 
\usepackage{amsmath,amssymb,amsthm}

\newtheorem{theorem}{Theorem}[section]

\newtheorem{lemma}[theorem]{Lemma}

\newtheorem{remark}{Remark}

\newcommand{\al}{\alpha}
\newcommand{\bt}{\beta}
\newcommand{\la}{\lambda}

\newcommand{\s}{\sigma}

\newcommand{\be}{\begin{equation}}
\newcommand{\ee}{\end{equation}}
\newcommand{\bea}{\begin{eqnarray}}
\newcommand{\eea}{\end{eqnarray}}
\newcommand{\no}{\nonumber}

\numberwithin{equation}{section}
\begin{document}
\title{\Large Asymptotics of the Hankel determinant and orthogonal polynomials arising from the information theory of MIMO systems}
\author{Chao Min\thanks{School of Mathematical Sciences, Huaqiao University, Quanzhou 362021, China; Email: chaomin@hqu.edu.cn}\: and Xiaoqing Wu\thanks{School of Mathematical Sciences, Huaqiao University, Quanzhou 362021, China}}


\date{\today}
\maketitle
\begin{abstract}
We consider the Hankel determinant and orthogonal polynomials with respect to the deformed Laguerre weight
$
w(x; t) = {x^\alpha }{\mathrm e^{ - x}}{(x + t)^\lambda },\; x\in \mathbb{R}^{+}
$
with parameters $\alpha > -1,\; t > 0$ and $\lambda \in \mathbb{R}$. This problem originates from the information theory of single-user multiple-input multiple-output (MIMO) systems studied by Chen and McKay [{\em IEEE Trans. Inf. Theory} {\bf 58} ({2012}) {4594--4634}]. By using the ladder operators for orthogonal polynomials with general Laguerre-type weights, we obtain a system of difference equations and a system of differential-difference equations for the recurrence coefficients $\alpha_n(t)$ and $\beta_n(t)$. We also show that the orthogonal polynomials satisfy a second-order ordinary differential equation. By using Dyson's Coulomb fluid approach, we obtain the large $n$ asymptotic expansions of the recurrence coefficients $\alpha_n(t)$ and $\beta_n(t)$, the sub-leading coefficient $\mathrm p(n, t)$ of the monic orthogonal polynomials,  the Hankel determinant $D_n(t)$ and the normalized constant $h_n(t)$ for fixed $t\in\mathbb{R}^{+}$.  We also discuss the long-time asymptotics of these quantities as $t\rightarrow\infty$ for fixed $n\in\mathbb{N}$. The large $n$ and large $t$ asymptotics of the above quantities are very important for the study of the asymptotics of the mutual information distribution and two fundamental quantities (the outage capacity and the error probability) for single-user MIMO systems.
\end{abstract}

$\mathbf{Keywords}$:  Deformed Laguerre weight; Hankel determinant; Orthogonal polynomials;

Ladder operators; Coulomb fluid; Large $n$ asymptotics; Long-time asymptotics.

$\mathbf{Mathematics\:\: Subject\:\: Classification\:\: 2020}$: 42C05, 33C45, 41A60.

\section{Introduction}
In the important work of Chen and McKay \cite{cm}, they studied the outage capacity and the error probability of multiple-input multiple-output (MIMO) antenna wireless communication systems.  It was found that the key quantity for studying the outage capacity and the error probability is the moment generating function of the mutual information of multiantenna Gaussian channels, which in turn gives rise to the Hankel determinant generated by the deformed Laguerre weight for single-user MIMO systems:
$$
D_{n}[w]=D_{n}(t):=\det(\mu_{i+j}(t))_{i,j=0}^{n-1},
$$
where $\mu_j(t)$ is the $j$th moment given by
$$
\mu_j(t):=\int_{0}^{\infty}x^jw(x;t)dx,\qquad j=0,1,2,\dots
$$
and
\begin{equation}\label{wei}
w(x;t) = {x^\alpha }{\mathrm e^{ - x}}{(x + t)^\lambda }, \qquad x\in \mathbb{R}^{+},
\end{equation}
with parameters $\alpha > 0,\; t > 0$ and $\lambda \in \mathbb{R}$. When $\la=0$, the above weight (\ref{wei}) is reduced to the classical Laguerre weight $w_0(x) = {x^\alpha }{\mathrm e^{ - x}},\; x\in \mathbb{R}^{+},$ and the corresponding Hankel determinant is denoted by $D_n[w_0]$.

By using the ladder operator approach \cite{Chen1997} for the monic orthogonal polynomials with respect to the weight (\ref{wei}), Chen and McKay \cite{cm} found that the logarithmic derivative of the Hankel determinant $D_{n}(t)$ satisfies the Jimbo-Miwa-Okamoto $\s$-form of the Painlev\'{e} V equation. Later, Basor and Chen \cite{Basora} proved that the diagonal recurrence coefficient, with a suitable change of variables, satisfies a particular Painlev\'{e} V equation.

It is a well-known fact that the Hankel determinant  $D_n(t)$ can be expressed as the product of the normalized constants for the monic orthogonal polynomials with respect to the weight function (\ref{wei}) \cite[(2.1.6)]{Ismail}
\begin{equation}\label{d_nt1}
{D_n}(t)=\prod_{j=0}^{n-1} {{h_j}(t)},
\end{equation}
where $h_j(t)>0$ is defined by the orthogonality condition
\begin{equation}\label{h_j}
h_j(t)\delta_{jk}=\int_{0}^{\infty}P_j(x;t)P_k(x;t)w(x;t)dx,\qquad j, k=0, 1, 2, \ldots.
\end{equation}
Here, $\delta_{jk}$ is the Kronecker delta, and $P_j(x; t)$ is the monic orthogonal polynomial of degree $j$ and it has the expansion
$$
P_n(x; t)=x^n+\mathrm{p}(n, t) x^{n-1}+\cdots,
$$
where $\mathrm{p}(n, t)$ is the sub-leading coefficient of $P_n(x; t)$ with the initial value $\mathrm{p}(0, t)=0$.

Classical orthogonal polynomials (such as Hermite, Laguerre and Jacobi polynomials) are orthogonal with respect to a weight $w(x)$ on the real line which satisfies the Pearson equation
\begin{equation}\label{r}
 \frac{d}{{dx}}(\sigma (x)w(x)) = \tau (x)w(x),
\end{equation}
where $\sigma (x)$ and $\tau (x)$ are polynomials with deg $\sigma \le 2$ and deg $\tau$=1. For semi-classical orthogonal polynomials, they have a weight $w(x)$ that satisfies the Pearson equation (\ref{r}), where $\sigma (x)$ and $\tau (x)$ are polynomials with deg $\sigma > 2$ or deg $\tau\neq 1$. See, e.g., \cite[Section 1.1.1]{W}.
It is not difficult  to find that (\ref{wei}) is a semi-classical weight, since it satisfies the Pearson equation (\ref{r}) with
$$
\sigma (x) = {x^2} + t x, \qquad \tau (x) =  - {x^2} + (2 + \alpha+ \lambda  - t)x + t (1+\alpha).
$$


One of the important properties of orthogonal polynomials is that they satisfy the three-term recurrence relation:
$$
xP_{n}(x;t)=P_{n+1}(x;t)+\al_n(t) P_n(x;t)+\beta_{n}(t)P_{n-1}(x;t),
$$
with the initial conditions
$$
P_0(x;t):=1,\qquad \beta_0(t) P_{-1}(x;t):=0.
$$
It can be found that the recurrence coefficients have the following expressions:
\begin{equation}\label{alpha1}
\alpha_{n}(t)=\mathrm{p}(n,t)-\mathrm{p}(n+1,t),
\end{equation}
\begin{equation}\label{be}
\beta_{n}(t)=\frac{h_{n}(t)}{h_{n-1}(t)}.
\end{equation}
Moreover, taking a telescopic sum of (\ref{alpha1}) produces an important identity
\begin{equation}\label{sum}
\sum_{j=0}^{n-1}\al_j(t)=-\mathrm{p}(n,t).
\end{equation}
Combining (\ref{be}) and (\ref{d_nt1}), we have the following relation:
\begin{equation}\label{bd}
\bt_n(t)=\frac{D_{n+1}(t) D_{n-1}(t)}{D_n^2(t)}.
\end{equation}
For more information about orthogonal polynomials, see \cite{Chihara,Ismail,Szego}.

The rest of the paper is organized as follows. In Section 2, we apply the ladder operators and compatibility conditions to orthogonal polynomials with respect to the deformed Laguerre weight. We obtain a discrete system and a system of differential-difference equations for the recurrence coefficients, and show that the orthogonal polynomials satisfy a second-order differential equation. The relations between the logarithmic derivative of the Hankel determinant, the sub-leading coefficient of the monic orthogonal polynomials and the recurrence coefficients have also been established. In Section 3, we use Dyson's Coulomb fluid approach to derive the large $n$ asymptotic expansions of the recurrence coefficients $\alpha_n(t)$ and $\beta_n(t)$, the sub-leading coefficient $\mathrm p(n, t)$, the Hankel determinant $D_n(t)$ and the normalized constant $h_n(t)$ for fixed $t\in\mathbb{R}^{+}$. In Section 4, we consider the long-time asymptotics ($t\rightarrow+\infty$) of the recurrence coefficients, the sub-leading  coefficient, the Hankel determinant and the normalized constant for fixed $n\in\mathbb{N}$.

\section{Ladder operators, compatibility conditions and difference and differential equations}
The ladder operator approach developed by Chen and Ismail \cite{Chen1997} is useful to analyze the recurrence coefficients of orthogonal polynomials and the associated Hankel determinants; see, e.g., \cite{Basora,cm,Dai,Filipuk,Min2021,Min2023}. However, they require that the weight be well defined at the endpoints of the orthogonality interval (vanish at the endpoints most of the time in practice). Recently, the first author and Fang \cite{F} derived the ladder operators and associated compatibility conditions for orthogonal polynomials with a general Laguerre-type weight of the form
\be\label{gl}
w(x) = {x^\alpha }w_0(x),\qquad x\in \mathbb{R}^{+},
\ee
where $\alpha>-1,\; w_0(x)$ is a continuously differentiable function defined on $[0,\infty)$ and all the moments
$
\int_{0}^{\infty}x^jw(x)dx,\; j=0,1,2,\dots
$
exist.

The monic orthogonal polynomials with respect to the general Laguerre-type weight (\ref{gl}) satisfy the lowering and raising operator equations \cite{F}
\begin{equation}\label{bnn1}
\left(\frac{d}{dx}+B_{n}(x)\right)P_{n}(x)=\beta_{n}A_{n}(x)P_{n-1}(x),
\end{equation}
\begin{equation}\label{bn2}
\left(\frac{d}{dx}-B_{n}(x)-\mathrm{v}'(x)\right)P_{n-1}(x)=-A_{n-1}(x)P_{n}(x),
\end{equation}
where $\mathrm{v}(x):=-\ln w(x)$ and
\begin{equation}\label{an}
A_{n}(x):=\frac{1}{x}\cdot\frac{1}{h_{n}}\int_{0}^{\infty}\frac{x \mathrm{v}'(x)-y \mathrm{v}'(y)}{x-y}P_{n}^{2}(y)w(y)dy,
\end{equation}
\begin{equation}\label{bn}
B_{n}(x):=\frac{1}{x} \left( \dfrac{1}{h_{n-1}}\int_{0}^{\infty}\dfrac{x \mathrm{v}'(x)-y \mathrm{v}'(y)}{x-y}P_{n}(y)P_{n-1}(y)w(y)dy-n\right).
\end{equation}

Based on the above definitions of  $A_n(x)$ and $B_n(x)$, the first author and Fang \cite{F} proved that $A_n(x)$ and $B_n(x)$ satisfy the following compatibility conditions:
\be
B_{n+1}(x)+B_{n}(x)=(x-\alpha_n) A_{n}(x)-\mathrm{v}'(x), \tag{$S_{1}$}
\ee
\be
1+(x-\alpha_n)(B_{n+1}(x)-B_{n}(x))=\beta_{n+1}A_{n+1}(x)-\beta_{n}A_{n-1}(x), \tag{$S_{2}$}
\ee
\be
B_{n}^{2}(x)+\mathrm{v}'(x)B_{n}(x)+\sum_{j=0}^{n-1}A_{j}(x)=\beta_{n}A_{n}(x)A_{n-1}(x), \tag{$S_{2}'$}
\ee
where ($S_{2}'$) is obtained  from a suitable combination of ($S_{1}$) and ($S_{2}$).

In this section, we apply the above ladder operators and compatibility conditions to orthogonal polynomials with a slightly more general weight compared to (\ref{wei}) studied by Chen and McKay \cite{cm}:
\begin{equation}\label{weig}
w(x;t) = {x^\alpha }{\mathrm e^{ - x}}{(x + t)^\lambda }, \qquad x\in \mathbb{R}^{+},
\end{equation}
with parameters $\alpha > -1,\; t > 0$ and $\lambda \in \mathbb{R}$. It is obvious that all the formulas in the Introduction still hold for the weight (\ref{weig}).

From (\ref{weig}) we have
\begin{equation}\label{pt1}
\mathrm v(x) =  - \ln w(x) = x - \alpha \ln x - \lambda \ln (x + t).
\end{equation}
It follows that
\begin{equation}\label{vp2}
\mathrm v'(x)=1-\frac{\alpha }{x}-\frac{\lambda}{x+t},
\end{equation}
and
\begin{equation}\label{vp1}
\frac{x \mathrm v^\prime(x) - y \mathrm v^\prime(y)}{x - y} = 1- \frac{\lambda t }{(x + t)(y + t)}.
\end{equation}
Inserting (\ref{vp1}) into the definitions of $A_n(x)$ and $B_n(x)$ in (\ref{an}) and (\ref{bn}), we obtain
\begin{equation}\label{anz1}
{A_n}(x) =
\frac{1- {R_n}(t)}{x} +\frac{{R_n}(t)}{x + t},
\end{equation}
\begin{equation}\label{bnz1}
{B_n}(x) = -\dfrac{n+{r_n}(t)}{x} + \dfrac{{r_n}(t)}{x + t},
\end{equation}
where
\be\label{Rn}
{R_n}(t): = \frac{\lambda }{{h_n}}\int_0^\infty  {\frac{{P^2_n}(x)w(x)}{x + t}}d x,
\ee
\be\label{rn}
{r_n}(t): = \frac{\lambda }{{h_{n - 1}}}\int_0^\infty  {\frac{{{P_n}(x){P_{n - 1}}(x)}w(x)}{{x + t}}} d x.
\ee
\begin{remark}
The expressions of $A_n(x)$ and $B_n(x)$ in (\ref{anz1}) and (\ref{bnz1}) coincide with those in the paper of Chen and McKay \cite{cm}, while they have used integration by parts to get the results.
\end{remark}
\begin{lemma}
The recurrence coefficients and auxiliary quantities satisfy the following conditions:
\begin{equation}\label{s12}
{r_n} + {r_{n + 1}} = \lambda  - {R_n}(t + {\alpha _n}),
\end{equation}
\begin{equation}\label{s11}
2n + 1 +\alpha+ {r_n} + {r_{n + 1}}={\alpha _n}(1 - {R_n}) ,
\end{equation}
\begin{equation}\label{s21}
{r_n^2} - \lambda {r_n} = {\beta _n}{R_n}{R_{n - 1}},
\end{equation}
\begin{equation}\label{s23}
{(n + {r_n})^2} + \alpha (n + {r_n}) = {\beta _n}(1 - {R_n})(1 - {R_{n - 1}}),
\end{equation}
\begin{equation}\label{s22a}
 n\lambda-2r_n^2-(  2 n + \alpha-\lambda+ t ){r_n} - t\sum\limits_{j = 0}^{n - 1} {{R_j}}  = {\beta _n}[ {R_n}(1 - {R_{n - 1}}) + {R_{n - 1}}(1 - {R_n} )].
\end{equation}
\end{lemma}
\begin{proof}
The results are obtained by substituting (\ref{anz1}) and (\ref{bnz1}) into the compatibility conditions ($S_{1}$) and ($S_{2}'$).
See also \cite{cm}.
\end{proof}

\begin{lemma}
The auxiliary quantities $R_n(t)$ and $r_n(t)$ can be expressed in terms of the recurrence coefficients $\alpha_n$ and $\beta_n$ as follows:
\begin{equation}\label{alpha2}
 t{R_n}(t)=2 n+1 +  \alpha  + \lambda-{\alpha _n},
\end{equation}
\begin{equation}\label{s25}
(2 n+\alpha  + \lambda )t{r_n}(t) = - n(n + \alpha )t-\beta _n(4 n + 2 \alpha  + 2\lambda-t-{\alpha _n}-{\alpha _{n-1}}).
\end{equation}
\end{lemma}
\begin{proof}
Eliminating $r_n+r_{n+1}$ from the combination of (\ref{s12}) and (\ref{s11}), we get (\ref{alpha2}). Making a difference of (\ref{s21}) and (\ref{s23}) gives
\begin{equation}\label{s24}
(2 n+\alpha  + \lambda )t{r_n}={\beta _n}(t-t{R_n} - t{R_{n - 1}}) - n(n + \alpha )t.
\end{equation}
Substituting (\ref{alpha2}) into (\ref{s24}), we obtain (\ref{s25}).
\end{proof}
\begin{theorem}\label{th1}
The recurrence coefficients $\alpha_n$ and $\bt_n$ satisfy the discrete system
\begin{subequations}\label{ds1}
\begin{align}
&\:(2 n+2+\alpha  + \lambda  )[n(n+\al) t+ \beta_{n }  (4 n + 2\alpha  + 2\lambda -t - \alpha _{n } - \alpha_{n-1})]\no\\
&+(2 n+\alpha  + \lambda  )[(n+1)(n+1+\al) t+ \beta_{n + 1}  (4 n +4+ 2\alpha  + 2\lambda -t - \alpha _{n + 1} - \alpha_n)]\no\\
&=(2 n+\alpha  + \lambda  )(2 n+2+\alpha  + \lambda  )[(t + \alpha _n )(2 n +1+\alpha  + \lambda - \alpha _n)-\lambda t],\label{d11}\\[8pt]
&\:[n(n+\al) t+ \beta_{n }  (4 n + 2\alpha  + 2\lambda -t - \alpha _{n } - \alpha_{n-1})]^2\no\\
&+(2n+\al+\la)\la t[n(n+\al) t+ \beta_{n }  (4 n + 2\alpha  + 2\lambda -t - \alpha _{n } - \alpha_{n-1})]\no\\
&=(2 n+\alpha  + \lambda  )^2\bt_n(2 n +1+\alpha  + \lambda - \alpha _n)(2 n -1+\alpha  + \lambda - \alpha _{n-1}).\label{d12}
\end{align}
\end{subequations}
\end{theorem}
\begin{proof}
The theorem is established by substituting (\ref{alpha2})  and (\ref{s25}) into (\ref{s12}) and (\ref{s21}), respectively.
\end{proof}
\begin{remark}
The above discrete system will play an important role in the derivation of the large $n$ asymptotic expansions  of $\alpha_n(t)$ and $\beta _n(t)$ in the next section.
\end{remark}

\begin{theorem}
The orthogonal polynomials $P_n(x)$ satisfy the second-order differential equation
\begin{align}\label{p4}
&{P''_n}(x) - \bigg(\mathrm v^\prime(x)  + \dfrac{{A_n'(x)}}{{{A_n}(x)}}\bigg)P_n'(x) + \bigg( B_n'(x) - B_n^2(x) - \mathrm v^\prime(x) {B_n}(x)\no\\
&+ {\beta _n}{A_n}(x){A_{n - 1}}(x)-\dfrac{{A_n'(x){B_n}(x)}}{{{A_n}(x)}}\bigg) {P_n}(x) = 0,
\end{align}
where $\mathrm v^\prime(x)$  is given by (\ref{vp2}), and $A_n(x)$ and $B_n(x)$ are expressed in terms of ${\alpha _n}$ and $\beta_n$ as follows:
\begin{equation}\label{A1}
A_n(x)=\frac{x+t-2n-1-\alpha-\lambda+\alpha_n}{x(x+t)} ,
\end{equation}
\begin{equation}\label{B1}
B_n(x)= \frac{n(n+\al)t-n(2n+\al+\la)(x+t)+\beta _n(4 n + 2 \alpha  +2\lambda-t-\alpha _n-\alpha _{n-1})}{(2 n+\alpha  + \lambda )x(x+t)} .
\end{equation}
Note that ${\alpha _n}$ and $\beta_n$ are independent of $x$.
\end{theorem}
\begin{proof}
Equation (\ref{p4}) is obtained by eliminating $P_{n-1}(x)$ from the ladder operator equations (\ref{bnn1}) and (\ref{bn2}). Substituting (\ref{alpha2}) and (\ref{s25}) into (\ref{anz1})  and (\ref{bnz1}), we get (\ref{A1}) and (\ref{B1}), respectively. This completes the proof.
\end{proof}
\begin{remark}
When $\la=0$, it can be found from (\ref{anz1})--(\ref{rn}) that
$$
R_n(t)=0,\qquad r_n(t)=0,
$$
and
$$
A_n(x)=\frac{1}{x},\qquad B_n(x)=-\frac{n}{x}.
$$
It follows from (\ref{alpha2}) and (\ref{s25}) that
$$
\al_n=2n+\al+1,\qquad \bt_n=n(n+\al).
$$
Then, the differential equation (\ref{p4}) is reduced to
$$
xP_n''(x)+(\al+1-x)P_n'(x)+nP_n(x)=0,
$$
which is the differential equation satisfied by the classical Laguerre polynomials \cite[p. 100]{Szego}.
\end{remark}

Next, we consider the time evolution of our problem.
From (\ref{h_j}) we have
$$
\int_{0}^{\infty}P_n(x;t)P_{n-1}(x;t)w(x;t) dx=0
$$
and
$$
\int_{0}^{\infty}P_n^2(x;t)w(x;t) dx =h_n(t).
$$
Taking derivatives with respect to $t$ produce
\begin{equation}\label{ddt1}
\frac{d}{d t}\mathrm{p}(n,t)=-r_n(t)
\end{equation}
and
\begin{equation}\label{dlnhnt}
\frac{d}{d t}\ln h_n(t)=R_n(t),
\end{equation}
respectively.
Using (\ref{alpha1}) and (\ref{be}), it follows that
\begin{equation}\label{all1}
\alpha_n'(t)=r_{n+1}(t)-r_n(t)
\end{equation}
and
\be\label{bn3}
\beta_{n}'(t)=\beta_n (R_{n}(t)- R_{n-1}(t)).
\ee
\begin{theorem}
The recurrence coefficients $\alpha_n(t)$ and $\beta_n(t)$ satisfy the coupled differential-difference equations:
\begin{subequations}\label{tb}
\begin{align}
&\:(2 n+\alpha+\lambda)(2 n+2+\alpha+\lambda) t \alpha'_n(t)=-t[2n^2+2n(1+\alpha+\lambda)+(1+\alpha)(\alpha+\lambda)]\no\\
&  -(2 n+\alpha+\lambda)\beta_{n+1}(4n+4+2\alpha+2\lambda-t-\alpha_{n+1}-\alpha_n)\no\\
&  +(2 n+2+\alpha+\lambda)\beta_{n}(4n+2\alpha+2\lambda-t-\alpha_n-\alpha_{n-1}),\label{tb2}\\[8pt]
&\:\:t \beta'_n(t)=\beta_n(\alpha_{n-1}-\alpha_n+2).\label{tb1}
\end{align}
\end{subequations}
\end{theorem}
\begin{proof}
Substituting (\ref{s25}) into (\ref{all1}) and (\ref{alpha2}) into (\ref{bn3}), we obtain the desired results.
\end{proof}

Let ${H_n}(t)$ denote the quantity related to the logarithmic derivative of the Hankel determinant $D_{n}(t)$, i.e.,
\begin{equation}\label{Pin}
{H _n}(t): = t\frac{d}{{dt}}\ln {D_n}(t).
\end{equation}
It follows from (\ref{d_nt1}) and (\ref{dlnhnt}) that
\begin{equation}\label{eq4}
{H _n}(t)=t \sum_{j=0}^{n-1}R_j(t)
\end{equation}
and
\begin{equation}\label{tR}
t R_n(t)={H_{n+1}}(t)-{H _n}(t).
\end{equation}
Using (\ref{sum}), (\ref{alpha2}) and (\ref{eq4}), we have
\be
{H _n}(t) =n(n + \alpha  + \lambda ) + \mathrm p(n,t).\label{ex1}
\ee
Replacing $\sum_{j=0}^{n-1}R_j(t)$ by ${H _n}(t)$ in (\ref{s22a}), and using (\ref{s21}) and (\ref{s23}), we obtain
\be
{H _n}(t)
=n(n + \alpha  + \lambda ) - {\beta _n} - t{r_n}(t). \label{ex2}
\ee
From (\ref{ex1}) and (\ref{ex2}), we have the relation
\begin{equation}\label{pp}
\mathrm{p}(n,t)= - {\beta _n} - t{r_n}(t).
\end{equation}
\begin{theorem}\label{thm2}
Let $W_n(t):=1-\frac{1}{1-R_n(t)}$, where $R_n(t)$ is closely related to the diagonal recurrence coefficient $\al_n$ by (\ref{alpha2}). Then $W_n(t)$ satisfies a particular Painlev\'{e} V equation \cite[p. 1]{Gromak}
$$
W_n''=\frac{3W_n-1}{2W_n(W_n-1)}(W_n')^2-\frac{W_n'}{t}+\frac{(W_n-1)^2}{t^2}\left(\mu_0 W_n+\frac{\mu_1}{W_n}\right)+\frac{\mu_2 W_n}{t}+\frac{\mu_3 W_n(W_n+1)}{W_n-1},
$$
with parameters $\mu_0=\frac{\al^2}{2}, \mu_1=-\frac{\la^2}{2}, \mu_2=2n+\al+\la+1, \mu_3=-\frac{1}{2}$.

Let $\s_n(t):=H_n(t)-n\la$. Then $\s_n(t)$ satisfies the Jimbo-Miwa-Okamoto $\s$-form of the Painlev\'{e} V equation \cite[(C. 45)]{Jimbo1981}
$$
(t\s_n'')^2=[\s_n-t\s_n'+2(\s_n')^2+(\nu_0+\nu_1+\nu_2+\nu_3)\s_n']^2-4(\nu_0+\s_n')(\nu_1+\s_n')(\nu_2+\s_n')(\nu_3+\s_n'),
$$
with parameters $\nu_0=0, \nu_1=-n, \nu_2=\la, \nu_3=-n-\al$.
\end{theorem}
\begin{proof}
See \cite{Basora} and \cite{cm}. We write the $\s$-form of the Painlev\'{e} V equation more clearly here.
\end{proof}
\begin{remark}
From the analysis in this section, it can be seen that all the finite $n$ results for the deformed Laguerre weight in \cite{Basora} and \cite{cm} can be extended from $\al>0$ to $\al>-1$.
\end{remark}

\section{Large $n$ asymptotics}
In random matrix theory (RMT), the joint probability density function of the eigenvalues of $ n \times n $
Hermitian matrices from the deformed  Laguerre unitary ensemble associated with the weight (\ref{weig}) is given by \cite{Mehta}
$$
p\left(x_1, x_2, \ldots, x_n\right) \prod_{k=1}^n d x_k=\frac{1}{Z_n(t)}  \prod_{1 \leq i<j \leq n}(x_i-x_j)^2 \prod_{k=1}^n w({x_k};t) d x_k,
$$
where $Z_n(t)$ is the partition function
$$
Z_n(t):= \int_{(0, \infty)^n} \prod_{1 \leq i<j \leq n}(x_i-x_j)^2 \prod_{k=1}^n w({x_k};t) d x_k.
$$
Note that there is only a constant difference between $D_{n}(t)$ and $Z_n(t)$ using Andr\'{e}ief's or Heine's integration formula \cite[(2.2.11)]{Szego}:
$$
D_n(t)=\frac{1}{n!}Z_n(t).
$$

According to Dyson's Coulomb fluid approach \cite{Dyson}, the eigenvalues of  Hermitian matrices from the unitary ensemble can be approximated as a continuous fluid with the equilibrium density $\sigma(x)$ when $n$ is sufficiently large. When $\al\geq 0$, the potential of our problem in (\ref{pt1}) is either convex or satisfies the condition that $x\mathrm{v}'(x)$ increases on $\mathbb{R}^{+}$. In this case, the equilibrium density $\sigma(x)$ is supported on a single interval, denoted by $(0,b) (b>0)$; see Chen and Ismail \cite{ci1997} and also \cite[p. 198--199]{Saff}. According to \cite{ci1997}, the equilibrium density $\sigma(x)$ can be determined by the constrained minimization problem:
\begin{equation}\label{con}
\mathop {\min }\limits_\sigma  F[\sigma ] \qquad \text{ subject to}\qquad  \int_{0}^{b}\s(x)dx=n,
\end{equation}
where $F[\s]$ is the free energy defined by
\begin{equation}\label{fe}
F[\s]:=\int_{0}^{b}\s(x)\mathrm{v}(x)dx-\int_{0}^{b}\int_{0}^{b}\s(x)\ln|x-y|\s(y)dxdy.
\end{equation}

Upon minimization, the equilibrium density $\sigma(x)$ is found to satisfy the integral equation
\begin{equation}\label{ie}
\mathrm{v}(x)-2\int_{0}^{b}\ln|x-y|\s(y)d y=A,\qquad x\in (0,b),
\end{equation}
where $A$ is the Lagrange multiplier, a constant independent of $x$.
Equation (\ref{ie}) is transformed into the following singular integral equation by taking a derivative with respect to $x$:
\begin{equation}\label{sie}
\mathrm{v}'(x)-2 P\int_{0}^{b}\frac{\sigma(y)}{x-y}d y=0,\qquad x\in (0,b),
\end{equation}
where $P$ denotes the Cauchy principal value. Multiplying by $\sqrt{\frac{x}{b-x}}$ on both sides of (\ref{sie}) and integrating from $0$ to $b$ with respect to $x$, we obtain
\begin{equation}\label{sup2}
\frac{1}{2\pi}\int_{0}^{b}\sqrt{\frac{x}{b-x}}\:\mathrm{v}'(x)d  x=n.
\end{equation}
The endpoint $b$ of the support of the equilibrium density is determined by equation (\ref{sup2}).

Taking a partial derivative with respect to $n$ in (\ref{fe}) and using the normalization condition in (\ref{con}) and (\ref{ie}), we have the relation
\begin{equation}\label{pd}
\frac{\partial F[\s]}{\partial n}=A.
\end{equation}
Furthermore, it was shown in \cite{ci1997} that the recurrence coefficients have the following asymptotic behavior as $n\rightarrow\infty$:
\begin{equation}\label{al}
\alpha_n=\frac{b}{2}+O\left(\frac{\partial^2 A}{\partial t \partial n}\right),
\end{equation}
\begin{equation}\label{beta}
\beta_n=\left(\frac{b}{4}\right)^2\left(1+O\left(\frac{\partial^3 A}{\partial n^3}\right)\right).
\end{equation}

Substituting (\ref{vp2}) into (\ref{sup2}), we obtain an important algebraic equation for $b>0$:
\be\label{eq}
b - 2\alpha  - 2\lambda  + 2\lambda\sqrt {\frac{t}{{b + t}}} = 4n.
\ee
The equation (\ref{eq}) has a unique positive solution and the large $n$ series expansion reads
\begin{align}
b =&\:4n + 2 (\alpha  +  \lambda)  - \dfrac{\lambda\sqrt t }{ n^{1/2} } + \dfrac{\lambda\sqrt t \left[t + 2\left(\alpha  + \lambda \right)\right]}{8 n^{3/2}}-\dfrac{\lambda ^2 t }{8 n^2}-\dfrac{{3 \lambda\sqrt t {{\left[ {t + 2\left(\alpha  + \lambda \right)} \right]}^2}}}{128 n^{5/2}}\no\\[8pt]
&+ \dfrac{\lambda ^2 t \left[t + 2\left( {\alpha  + \lambda } \right)\right]}{16{n^3}} + O(n^{ - 7/2}).\no
\end{align}
It follows that
\begin{align}\label{b2}
\dfrac{b}{2}=&\:2n +  \alpha  +  \lambda  - \dfrac{\lambda\sqrt t }{ 2n^{1/2} } + \dfrac{\lambda\sqrt t \left[t + 2\left(\alpha  + \lambda \right)\right]}{16 n^{3/2}}-\dfrac{\lambda ^2 t }{16 n^2}-\dfrac{{3 \lambda\sqrt t {{\left[ {t + 2\left(\alpha  + \lambda \right)} \right]}^2}}}{256 n^{5/2}}\no\\[8pt]
&+ \dfrac{\lambda ^2 t \left[t + 2\left( {\alpha  + \lambda } \right)\right]}{32{n^3}} + O(n^{ - 7/2})
\end{align}
and
\begin{align}\label{b4}
\left(\dfrac{b}{4}\right)^2=&\:n^2+ n(\alpha +\lambda) -\dfrac{\lambda \sqrt{t}\: n^{1/2}}{2}+\dfrac{(\alpha +\lambda)^2}{4}+\dfrac{\lambda  \sqrt{t} [t-2 (\alpha +\lambda)]}{{16 n^{1/2}}}\no\\[8pt]
&-\dfrac{\lambda \sqrt{t} \left[3 t^2+4 t (\alpha +\lambda ) -4 (\alpha +\lambda )^2\right]}{{256{n^{3/2}}}}+\dfrac{\lambda ^2 t^2}{64 n^2}+O(n^{ - 5/2}).
\end{align}

Multiplying by $\frac{1}{\sqrt{(b-x)x}}$ on both sides of (\ref{ie}) and integrating with respect
to $x$ from $0$ to $b$, and following the similar computations in \cite[p. 406--407]{Min2021}, we find the Lagrange multiplier $A$ is
\begin{align}
A =& \dfrac{1}{\pi }\int_0^b {\dfrac{\mathrm{v}(x)}{\sqrt {(b - x)x}}} dx - 2n\ln\frac{b}{4}\no\\
            =&\frac{b}{2}-(2n+\alpha) \ln \frac{b}{4}-\la\ln \frac{t}{4}   - 2\lambda \ln \left( {\sqrt {\dfrac{b + t}{t}}  + 1} \right).\no
\end{align}
The large $n$ series expansion reads
\begin{align}\label{a}
A =&- 2n \ln n + 2n- (\alpha  + \lambda ) \ln n - \dfrac{\lambda \sqrt t}{ n^{1/2}} - \dfrac{(\alpha +\lambda )^2}{4n} + \dfrac{\lambda  \sqrt{t}\: (6 \alpha +6 \lambda +t)}{24 n^{3/2}}\no\\[8pt]
& + \dfrac{2 \alpha ^3 + 6 \alpha ^2\lambda  + 3(2\alpha-t) \lambda ^2 + 2 \lambda ^3}{48 n^2} + O(n^{- 5/2}).
\end{align}
From (\ref{pd}) and (\ref{a}), we see that the free energy $F[\s]$ has the large $n$ asymptotic expansion
\begin{small}
  \begin{align}\label{F}
F[\s]=&-n^2 \ln n+\dfrac{3}{2} n^2-(\alpha +\lambda )n \ln n +(\alpha +\lambda) n -2 \lambda  \sqrt t\: n^{1/2} -\dfrac{(\alpha +\lambda )^2 \ln n}{4} +C\no\\[8pt]
&-\dfrac{\lambda  \sqrt{t}\: (6 \alpha +6 \lambda +t)}{12 n^{1/2}}-\dfrac{2 \alpha ^3 + 6 \alpha ^2\lambda  + 3(2\alpha-t) \lambda ^2 + 2 \lambda ^3}{48 n}+O(n^{-3/2}),
\end{align}
\end{small}
where $C$ is an undetermined constant independent of $n$.

With these ingredients in hand, we are now ready to derive the asymptotic expansions of the recurrence coefficients as $n\rightarrow\infty$ for fixed $t>0$.
\begin{theorem}\label{thrc}
The recurrence coefficients $\alpha_n(t)$ and $\beta_{n}(t)$ have the following large $n$ expansions:
\begin{align}
\alpha_n(t)=&\:2 n+1+\alpha +\lambda -\dfrac{\lambda  \sqrt{t}}{2 n^{1/2}}+\dfrac{\lambda \left[4 t^2+8 t (\alpha +\lambda +1)+4 \alpha ^2-1\right]}{64 \sqrt{t}\: n^{3/2}}\no\\[5pt]
&-\dfrac{\lambda ^2 \left(4 t^2-4 \alpha ^2+1\right)}{64 t n^2}+O(n^{-5/2}),\label{aa}\\[8pt]
\beta_n(t)=&\:n^2 +(\alpha +\lambda )n-\dfrac{\lambda \sqrt{t}\:n^{1/2}}{2} +\dfrac{\lambda  (2 \alpha +\lambda )}{4} +\dfrac{\lambda  \left[4 t^2-8 t (\alpha +\lambda )-12 \alpha ^2+3\right]}{64 \sqrt{t}\: n^{1/2} }  \no\\[5pt]
&+\dfrac{\lambda ^2 \left(1-4 \alpha ^2\right) }{32 t n}+O(n^{-3/2}).\label{ab}
\end{align}
\end{theorem}
\begin{proof}
Taking account of (\ref{al}), (\ref{beta}), (\ref{b2}),  (\ref{b4}) and (\ref{a}), we find that $\al_n$ and $\bt_{n}$ have the large $n$ expansions of the forms
\begin{equation}\label{al1}
\alpha _n \sim a_{- 2} n + a_{- 1} n^{1/2} + {a_0} + \sum\limits_{j = 1}^\infty  {\frac{a_j}{n^{j/2}}},
\ee
\be\label{be1}
{\beta _n} \sim b_{- 4} n^2+ b_{- 3} n^{3/2} + b_{- 2} n +b_{- 1}n^{1/2} + {b_0} + \sum\limits_{j = 1}^\infty  {\frac{b_j}{n^{j/2}}},
\end{equation}
where
$$
{a_{- 2}} = 2,\qquad\qquad b_{-4}=1.
$$
Substituting (\ref{al1}) and (\ref{be1}) into the discrete system satisfied by the recurrence coefficients in (\ref{ds1}), and taking a large $n$ limit, we obtain the expansion coefficients $a_j$ and  $b_j$ recursively by equating coefficients of powers of $n$ on both sides. The first few terms are
$$
\begin{aligned}
&a_{-1}=0,\qquad b_{-3}=0;\qquad  a_0=1 + \alpha  + \lambda ,\qquad  b_{-2}= \alpha  + \lambda;\\
&a_1= - \frac{\lambda \sqrt t}{2},\qquad {b_{ - 1}} =  - \frac{\lambda \sqrt t }{2};\qquad a_2=0,\qquad b_0=\frac{\lambda \left( 2\alpha  + \lambda \right)}{4};\\
&a_3= \frac{\lambda \left[4 t^2 + 8t\left(\alpha  + \lambda+1 \right)+ 4 \alpha ^2-1\right]}{64\sqrt t },\qquad b_1=\frac{\lambda \left[ 4 t^2 - 8t\left( \alpha  + \lambda \right)- 12 \alpha ^2+3\right]}{64\sqrt t};\\
&a_4= - \frac{\lambda ^2 \left( 4 t^2 - 4 \alpha ^2+1\right)}{64 t},\qquad b_2=\frac{\lambda ^2 \left(1-4 \alpha ^2\right)}{32 t}.
\end{aligned}
$$
This completes the proof.
\end{proof}

Next, based on the large $n$ asymptotic expansions of the recurrence coefficients, we continue to derive the large $n$ expansions of the sub-leading coefficient $\mathrm p(n,t)$, the quantity $H_n(t)= t\frac{d}{{dt}}\ln {D_n}(t)$, the Hankel determinant $D_n(t)$ and the normalized constant $h_n(t)$ for fixed $t>0$. Note that the asymptotic expansions of the recurrence coefficients in Theorem \ref{thrc} are singular at the point $t=0$, which leads to the fact that all the other quantities discussed in the following have the same feature.

\begin{theorem}
The large $n$ asymptotic expansion of $\mathrm{p}(n,t)$ is given by
\begin{align}\label{pna}
\mathrm{p}(n,t)=&- {n^2} - (\alpha  + \lambda )n + \lambda \sqrt t\: n^{1/2}  - \dfrac{\la(2t+2\al+\la)}{4} + \dfrac{\lambda \left[4t^2+8(\al+\la)t+ 4 \alpha^2 - 1\right]}{32\sqrt t \: n^{1/2}}\no\\
&- \dfrac{\lambda ^2\left(4 t^2- 4 \alpha ^2+1\right)}{64 t n} + O\left( n^ {- 3/2}\right).
\end{align}
\end{theorem}
\begin{proof}
Combining (\ref{s25}) and (\ref{pp}), we can express $\mathrm{p}(n,t)$ in terms of the recurrence coefficients:
$$
\mathrm{p}(n,t)=\frac{n(n+\al)t+\bt_n(2n+\al+\la-t-\al_n-\al_{n-1})}{2n+\al+\la}.
$$
Substituting (\ref{aa}) and (\ref{ab}) into the above and taking a large $n$ limit, we obtain the result in (\ref{pna}).
\end{proof}
\begin{theorem}\label{hnt}
The quantity ${H _n}(t)= t\frac{d}{{dt}}\ln {D_n}(t)$ has the following expansion as $n\rightarrow\infty$:
\begin{align}\label{dd}
{H _n}(t)=&\:\lambda \sqrt t\: n^{1/2}  - \dfrac{\la(2t+2\al+\la)}{4} + \dfrac{\lambda \left[4t^2+8(\al+\la)t+ 4 \alpha^2 - 1\right]}{32\sqrt t \: n^{1/2}}- \dfrac{\lambda ^2\left(4 t^2- 4 \alpha ^2+1\right)}{64 t n} \no\\
&+ O\left( n^ {- 3/2}\right).
\end{align}
\end{theorem}
\begin{proof}
Substituting (\ref{pna}) into (\ref{ex1}) gives the desired result.
\end{proof}
Since the above asymptotic expansion of $H_n(t)$ is singular at $t=0$, one can not derive the large $n$ asymptotic expansion of $\ln D_n(t)$ by integrating $H_n(t)/t$ from 0 to $t$. The derivation of the large $n$ asymptotics of the Hankel determinant will be more difficult.
\begin{theorem}\label{thm1}
The Hankel determinant $D_n(t)$ has the following asymptotic expansion as $n\rightarrow\infty$:
\begin{align}\label{dnt1}
\ln D_n(t)=  &\: {n^2}\ln n- \dfrac{3n^2}{2} + \left(\alpha  + \lambda\right)n\ln n- [\alpha+\lambda-\ln (2\pi)]n+ 2 \lambda \sqrt t\:   n^{1/2}\no\\[8pt]
&+ \dfrac{ 6\alpha ^2+6\alpha \lambda  + 3 \lambda ^2-2 }{12}\ln n-c_0+ \dfrac{\lambda \left[4 t^2 + 24 t\left(\alpha  + \lambda\right)- 12 \alpha ^2+3 \right]}{48\sqrt t \:   n^{1/2}} \no\\[8pt]
&- \dfrac{12 \lambda ^2 t^2 - 8 t(\alpha +\lambda ) \left(4 \alpha ^2+2 \alpha  \lambda +\lambda ^2-2\right)+ 3 \lambda ^2 \left(4 \alpha ^2-1\right)}{192 t n}
+O(n^{-3/2}),
\end{align}
where the constant term reads
$$
-c_0=2 \zeta'(-1)- \ln G( \alpha+1 )+ \dfrac{\alpha}{2}\ln (2\pi)-\frac{\la}{4}\left[2t+\la\ln 4+(2\al+\la)\ln t\right],
$$
where $\zeta'(\cdot)$ is the derivative of the Riemann zeta function and $G(\cdot)$ is the Barnes $G$-function \cite{ew}.
\end{theorem}
\begin{proof}
Let
$$
F_n(t):=-\ln D_n(t)
$$
be the ``free energy''. Chen and Ismail \cite{ci1997} showed that $F_n(t)$ can be approximated by the free energy $F[\s]$ in (\ref{fe}) for sufficiently large $n$. Chen, Ismail and Van Assche \cite{ChenIsmail1998} also stated that the approximation is very accurate and effective, and applied the approach to three different examples. More recently, Basor, Chen and Haq \cite[Section 7.2]{Basor2015} used this approach to derive the large $n$ asymptotic expansion of the Hankel determinant for a modified Jacobi weight and also used other methods to show the results are consistent. In view of (\ref{F}), we assume that $F_n(t)$ has the large $n$ expansion form:
\be\label{fna}
{F_n}(t) =  {c_7}{n^2}\ln n + {c_6}n\ln n+ {c_5}\ln n+\sum_{j =  -\infty}^{4} {c_{j}}{n^{j/2}},
\ee
where $c_j,\;j=7, 6, 5, \ldots$ are the expansion coefficients to be determined.
Using (\ref{bd}) we have
\begin{equation}\label{dc}
\ln\bt_n=2 F_n(t)-F_{n+1}(t)-F_{n-1}(t).
\end{equation}

Substituting (\ref{ab}) and (\ref{fna}) into equation (\ref{dc}) and taking a large $n$ limit, we obtain the expansion coefficients $c_j$ (except $c_2$ and $c_0$) by equating coefficients of powers of $n$ on both sides.
The asymptotic expansion of $F_n(t)$ as $n\rightarrow\infty$ is then given by
\begin{align}
{F_n}(t) =  &- {n^2}\ln n - \left(\alpha  + \lambda\right)n\ln n- \dfrac{ 6\alpha ^2+6\alpha \lambda  + 3 \lambda ^2-2 }{12}\ln n+ \dfrac{3n^2}{2} + {c_2}n\no\\[8pt]
&- 2 \lambda \sqrt t\:   n^{1/2}+ {c_0}- \dfrac{\lambda \left[4 t^2 + 24 t\left(\alpha  + \lambda\right)- 12 \alpha ^2+3 \right]}{48\sqrt t \:   n^{1/2}} \no\\[8pt]
&+ \dfrac{12 \lambda ^2 t^2 - 8 t(\alpha +\lambda ) \left(4 \alpha ^2+2 \alpha  \lambda +\lambda ^2-2\right)+ 3 \lambda ^2 \left(4 \alpha ^2-1\right)}{192 t n}
+O(n^{-3/2}),\no
\end{align}
where ${c_2}$ and ${c_0}$ are undetermined. It follows that
\begin{align}
\ln D_n(t)=  &\: {n^2}\ln n - \dfrac{3n^2}{2}+ \left(\alpha  + \lambda\right)n\ln n- {c_2}n+ 2 \lambda \sqrt t\:   n^{1/2}\no\\[8pt]
&+ \dfrac{ 6\alpha ^2+6\alpha \lambda  + 3 \lambda ^2-2 }{12}\ln n- {c_0}+ \dfrac{\lambda \left[4 t^2 + 24 t\left(\alpha  + \lambda\right)- 12 \alpha ^2+3 \right]}{48\sqrt t \:   n^{1/2}} \no\\[8pt]
&- \dfrac{12 \lambda ^2 t^2 - 8 t(\alpha +\lambda ) \left(4 \alpha ^2+2 \alpha  \lambda +\lambda ^2-2\right)+ 3 \lambda ^2 \left(4 \alpha ^2-1\right)}{192 t n}
+O(n^{-3/2}).\no
\end{align}

To determine the two constants $c_2$ and $c_0$, we use a result from the work of Chen and Lawrence \cite{Chen1998}. Following \cite{Chen1998}, we have the large $n$ leading term asymptotics
$$
\frac{D_n[w]}{D_n[w_0]}\sim \exp\left(-\lambda \int_{\tilde{a}}^{\tilde{b}} f(x)\widetilde{\sigma}_{0}(x)dx-\frac{\lambda}{2}\int_{\tilde{a}}^{\tilde{b}} f(x)\varrho(x)dx\right),
$$
where
$$
\widetilde{\sigma}_{0}(x)=\frac{\sqrt{(\tilde{b}-x)(x-\tilde{a})}}{2\pi^{2}}P\int_{\tilde{a}}^{\tilde{b}}\frac{\mathrm{v}_{0}'(y) dy}{(y-x)\sqrt{(\tilde{b}-y)(y-\tilde{a})}}
$$
and
$$
\varrho(x)=\frac{\lambda}{2\pi^{2}\sqrt{(\tilde{b}-x)(x-\tilde{a})}}P\int_{\tilde{a}}^{\tilde{b}}\frac{\sqrt{(\tilde{b}-y)(y-\tilde{a})}}{y-x}f'(y)dy,
$$
and $\mathrm{v}_{0}(x)=-\ln w_0(x)=x-\al\ln x,\; f(x)=-\ln (x+t),\;x\in \mathbb{R}^{+}$ for our problem. The endpoints $\tilde{a}$ and $\tilde{b}$ are determined by the following two equations:
$$
\int_{\tilde{a}}^{\tilde{b}}\frac{\mathrm{v}_{0}'(x)}{\sqrt{(\tilde{b}-x)(x-\tilde{a})}}dx=0,
$$
$$
\frac{1}{2\pi}\int_{\tilde{a}}^{\tilde{b}}\frac{x\:\mathrm{v}_{0}'(x)}{\sqrt{(\tilde{b}-x)(x-\tilde{a})}}dx=n.
$$

It can be found that
$$
\widetilde{\sigma}_{0}(x)=\frac{\sqrt{(\tilde{b}-x)(x-\tilde{a})}}{2 \pi  x },
$$
$$
\varrho(x)=\frac{\lambda}{2 \pi  \sqrt{(\tilde{b}-x)(x-\tilde{a})}}\left(1-\frac{\sqrt{(\tilde{a}+t) (\tilde{b}+t)}}{x+t}\right),\\[5pt]
$$
and $\tilde{a}=2 n+\alpha -2 \sqrt{n(n+\alpha)  },\;\; \tilde{b}=2 n+\alpha +2 \sqrt{n(n+\alpha)  }$. After some elementary computations, we obtain
\begin{align}\label{ra}
\frac{D_n[w]}{D_n[w_0]}\sim &\exp\Big[\la n \ln n-\la n+2\la \sqrt{t}\:n^{\frac{1}{2}}+\frac{\la}{4}(2\al+\la)\ln n\no\\
&-\frac{\la}{4}\big(2t+\la\ln 4+(2\al+\la)\ln t\big)+o(1)\Big].
\end{align}

For the classical Laguerre weight $w_0(x) = {x^\alpha }{\mathrm e^{ - x}},\;x\in \mathbb{R}^{+},\alpha > -1$, it is known that the corresponding Hankel determinant has the closed-form expression \cite[p. 321]{Mehta}
\begin{align}
D_n[w_0]&=\det\left(\int_{0}^{\infty}x^{i+j}{x^\alpha }{\mathrm e^{ - x}}dx\right)_{i,j=0}^{n-1}\no\\
&=\frac{1}{n!}\prod_{j=1}^{n}\Gamma(j+1)\Gamma(j+\al)\no\\
&=\frac{G(n+1)G(n+\al+1)}{G(\al+1)},\no
\end{align}
where $G(z)$ is the Barnes $G$-function that satisfies the relation \cite{ew,Voros}
$$
G(z+1)=\Gamma(z)G(z),\qquad G(1)=1
$$
and it has the asymptotic expansion as $z\rightarrow+\infty$:
$$
\ln G(z+1 ) =z ^2 \left(\frac{\ln z }{2} - \frac{3}{4}\right) +  \frac{z}{2} \ln (2\pi)- \frac{\ln {z}}{12} + \zeta'(-1) + O(z ^{ - 2}),\qquad z  \to  + \infty,
$$
where $\zeta'(\cdot)$ is the derivative of the Riemann zeta function. It follows that
\begin{align}\label{as2}
\ln D_n[w_0] = &\:{n^2}\ln n - \dfrac{3}{2}{n^2}+ \alpha n \ln n- [ \alpha-\ln (2\pi)]n+ \dfrac{3\alpha ^2- 1}{6}\ln n\no\\
 &+ \dfrac{\alpha}{2}\ln (2\pi)+ 2 \zeta'(-1)- \ln G( \alpha+1 ) +O(n^{-1}).
\end{align}
Combining (\ref{ra}) and (\ref{as2}) yields
$$
-c_2=-\alpha-\lambda+\ln (2\pi)
$$
and
$$
-c_0=2 \zeta'(-1)- \ln G( \alpha+1 )+ \dfrac{\alpha}{2}\ln (2\pi)-\frac{\la}{4}\left[2t+\la\ln 4+(2\al+\la)\ln t\right].
$$
This completes the proof.
\end{proof}
\begin{remark}
The expansion form in (\ref{fna}) can be justified from the combination of (\ref{ra}) and (\ref{as2}). In addition, by taking a derivative with respect to $t$ in (\ref{dnt1}) and substituting it into (\ref{Pin}), one would find that the result is consistent with (\ref{dd}) in Theorem \ref{hnt}.
\end{remark}
\begin{remark}
Recently, Charlier and Gharakhloo \cite{CG} studied the large $n$ asymptotics of Hankel determinants for Laguerre-type and Jacobi-type weights. Note that
the weight of Charlier and Gharakhloo \cite[(1.2)]{CG} is a varying weight with the potential $n V(x)$, and our weight in (\ref{weig}) is independent of $n$. Therefore, our weight is not a special case of \cite[(1.2)]{CG} and can not be transformed to their case either. As a result, the large $n$ expansion form of our Hankel determinant in (\ref{dnt1}) is different from \cite[Theorem 1.2]{CG} since we have the fractional powers of $n$.
\end{remark}
\begin{theorem}
The normalized constant $h_n(t)$ has the large $n$ asymptotic expansion
\begin{align}
\ln h_n(t)=&\:2n\ln n - 2n + (1 + \alpha  + \lambda )\ln n + \ln(2\pi) + \dfrac{{\lambda \sqrt t }}{n^{1/2}} \no\\[8pt]
&+\frac{6\al(\al+\la+1)+3\la(\la+2)+2}{{12n}}+O(n^{-3/2}).\no
\end{align}
\end{theorem}
\begin{proof}
Using (\ref{d_nt1}), we have
\begin{align}\label{hn}
 \ln h_n(t)=\ln D_{n+1}(t)-\ln D_n(t).
\end{align}
Substituting (\ref{dnt1}) into (\ref{hn}) and taking a large $n$ limit, we obtain the desired result.
\end{proof}

At the end of this section, we would like to mention that the large $n$ asymptotic expansion of the Hankel determinant in Theorem \ref{thm1} can be used to study the asymptotics of the moment generating function of the mutual information for single-user MIMO systems when the number of antennas becomes large. Then, the asymptotics of the outage probability and the upper bound on the error probability can be calculated via the asymptotics of the moment generating function; see \cite[Section III]{cm}.

\section{Long-time asymptotics}
In this section, we consider the long-time asymptotics of the  recurrence coefficients $\alpha_n(t)$ and $\beta_n(t)$, the quantity $H_n(t)= t\frac{d}{{dt}}\ln {D_n}(t)$, the sub-leading  coefficient $\mathrm{p}(n,t)$, the Hankel determinant $D_n(t)$ and the normalized constant $h_n(t)$ when $t  \to  + \infty$ and $n$ is fixed.
To derive the main results, we first give some relations between the recurrence coefficients $\al_n(t), \bt_n(t)$ and the quantity $H_n(t)$ in the following lemma.
\begin{lemma}
We have the relations
\be\label{aa1}
\alpha_n(t)=2n+1+\alpha+\lambda+{H_n}(t)-{H _{n+1}}(t),
\ee
\be\label{bb1}
\beta_n(t)=n(n+\alpha+\lambda)-{H _n}(t)+t{H '_n}(t),
\ee
\be\label{an+1}
\alpha_{n}(t)=\alpha_{n-1}(t)+2-t\frac{d}{dt}\ln \beta_{n}(t).
\ee
\end{lemma}
\begin{proof}
Equating the right hand sides of (\ref{alpha2}) and (\ref{tR}), we have (\ref{aa1}). From (\ref{ddt1}) and (\ref{ex1}), we find
\be\label{dphi}
\dfrac{d}{dt}{H _n}(t)=\dfrac{d}{dt}\mathrm{p}(n,t)=-r_n(t).
\ee
The combination of (\ref{ex2}) and(\ref{dphi}) gives (\ref{bb1}). Finally, equation (\ref{an+1}) is from (\ref{tb1}).
\end{proof}
\begin{theorem}\label{th1t}
As $t  \to  + \infty$, the recurrence coefficients $\alpha_n(t)$ and $\beta_n(t)$ have the asymptotic expansions
\begin{align}
\alpha_n(t)=&\:2 n+ \alpha+ 1+ \dfrac{\lambda( 2 n+ \alpha +1)}{t}-\dfrac{\la\left[ 6 n^2+2n (3 \alpha - \lambda +3)+(\alpha +1) (\alpha -\lambda +2) \right]}{t^2}\no\\
&+O(t^{-3}),\label{d21}\\
\beta _n(t) =&\:n\left( {n + \alpha } \right) + \dfrac{{2\lambda n \left( {n + \alpha } \right) }}{t} - \dfrac{{3 \lambda n \left( {n + \alpha } \right)\left( {2n + \alpha  - \lambda } \right)}}{{{t^2}}}
+O(t^{-3}).\label{d22}
\end{align}
As $t  \to  + \infty$, the quantity $H_n(t)= t\frac{d}{{dt}}\ln {D_n}(t)$ has the asymptotic expansion
\be\label{Phi2}
{H _n}(t) =\lambda n  - \dfrac{{\lambda n  \left( {n + \alpha } \right) }}{t} + \dfrac{{\lambda n
 \left( {n + \alpha } \right)\left( {2n + \alpha  - \lambda } \right) }}{{{t^2}}}
+O(t^{-3}).
\ee
\end{theorem}
\begin{proof}
From the fact that $D_0(t)=1$, we have
\be\label{h0}
H_0(t)=0,
\ee
and it follows from (\ref{bb1}) that
$$
\bt_0(t)=0.
$$
Note that
\begin{align}\label{mu0}
{\mu _0}(t)&=\int_0^\infty  {{x^\alpha }} {\mathrm{e}^{ - x}}{(x + t)^\lambda }dx\no\\
&=t^{ \alpha+\lambda+1}\Gamma(\alpha+1)U(\alpha+1,  \alpha+\lambda+2, t),
\end{align}
where $U(a, b, z)$ is the Kummer function and it can be expressed in terms of the confluent hypergeometric function as \cite[Chapter 13]{Olver}
$$
U(a, b, z)=\frac{\Gamma(1-b)}{\Gamma(a-b+1)}\ _1F_1(a;b;z)+\frac{\Gamma(b-1)}{\Gamma(a)}z^{1-b}\ _1F_1(a-b+1;2-b;z).
$$
It follows that
$$
H_1(t)=t\frac{d}{dt}\ln D_1(t)=t\frac{d}{dt}\ln \mu_0(t)=\al+\la+1-\frac{t(\al+1)U(\alpha+2,  \alpha+\lambda+3, t)}{U(\alpha+1,  \alpha+\lambda+2, t)}
$$
and as $t  \to  + \infty$,
\be\label{h1}
{H _1}(t)
=\lambda  - \dfrac{{\lambda\left( {\al+1 } \right) }}{t} + \dfrac{{\lambda\left( {\al+1 } \right)\left( { \alpha  - \lambda+2 } \right) }}{{{t^2}}} +O(t^{-3}),
\ee
which is (\ref{Phi2}) with $n=1$.
By making use of (\ref{h0}) and (\ref{h1}), we find from (\ref{aa1}) and (\ref{bb1}) that
\be\label{a0}
  {\alpha _{0}(t)} =  \alpha+1  + \dfrac{{ \lambda\left( {\alpha+1 } \right) }}{t} - \dfrac{{ \lambda\left( { \alpha+1 } \right)\left( { \alpha  - \lambda+2 } \right) }}{{{t^2}}} + O(t^{-3})
\ee
and
\be\label{b1}
  {\beta _{1}(t)} =\alpha+1 + \dfrac{{2 \lambda \left( {\alpha+1 } \right) }}{t} - \dfrac{3 \lambda  {\left( {\alpha+1 } \right)\left( { \alpha  - \lambda+2 } \right) }}{t^2} +O(t^{-3}),
\ee
which are (\ref{d21}) and (\ref{d22}) with $n=0$ and $n=1$, respectively.
Using (\ref{a0}) and (\ref{b1}), it follows from (\ref{an+1}) that
$$
{\alpha _{1}(t)} = \alpha +3+\frac{\lambda(\alpha +3) }{t} -\frac{\lambda  \left[(\al+2)(\al+7)-\la(\al+3)\right]}{t^2} + O(t^{-3}),
$$
which is (\ref{d21}) with $n=1$.

Now we can use the mathematical induction. Suppose that (\ref{d21}), (\ref{d22}) and (\ref{Phi2}) are true, then using (\ref{aa1}) we have
\begin{align}
{H _{n+1}}(t)=&\:2n+1+\alpha+\lambda+{H_n}(t)-\alpha_n(t)\no\\
=&\:\lambda\left( { n+1} \right)  - \dfrac{{\lambda\left( { n+1} \right) \left( {n + \alpha+1 } \right)}}{t} + \dfrac{{\lambda\left( {n+1} \right)\left( { n+ \alpha+1 } \right)\left( { 2 n + \alpha  - \lambda+2 } \right) }}{{{t^2}}}\no\\
&+O(t^{-3}),\no
\end{align}
which is (\ref{Phi2}) with $n\mapsto n+1$. It follows from (\ref{bb1}) that
\begin{align}\label{bn1}
 \beta_{n+1}(t)=&\:(n+1)(n+\alpha+\lambda+1)-{H _{n+1}}(t)+t{H '_{n+1}}(t)\no\\
 =&\:\left( { n+1} \right)\left( { n + \alpha+1 } \right) + \dfrac{{2\lambda\left( {n+1} \right)\left( { n+ \alpha+1 } \right) }}{t} - \dfrac{{3\lambda {\left( {n+1} \right)\left( { n+ \alpha+1 } \right)\left( { 2 n+ \alpha  - \lambda +2 } \right)} }}{{{t^2}}}\no\\
&+O(t^{-3}),
\end{align}
which is (\ref{d22}) with $n\mapsto n+1$. Then, using (\ref{d21}) and (\ref{bn1}), we find from (\ref{an+1}) that
\begin{align}
\alpha_{n+1}(t)=&\:\alpha_{n}(t)+2-t\frac{d}{dt}\ln \beta_{n+1}(t)\no\\
=&\:2n + \alpha+3  + \dfrac{{\lambda\left( {2n +\alpha+ 3 } \right) }}{t} -\dfrac{\la\left[ 6 (n+1)^2+2(n+1) (3 \alpha - \lambda +3)+(\alpha +1) (\alpha -\lambda +2) \right]}{t^2} \no\\
&+O(t^{-3}),\no
\end{align}
which is (\ref{d21}) with $n\mapsto n+1$. Hence, the theorem is proved by induction.
\end{proof}
\begin{remark}
Recently, Clarkson and Jordaan \cite{Clarkson4} studied the generalized Airy polynomials and derived the long-time asymptotics of the recurrence coefficients. Their derivation is based on the Toda system satisfied by the recurrence coefficients, which is different from our problem.
\end{remark}
Next, we consider the long-time asymptotics of the sub-leading  coefficient $\mathrm{p}(n,t)$, the  Hankel determinant $D_n(t)$ and the normalized constant $h_n(t)$ from the above theorem.
\begin{theorem}
As $t  \to  + \infty$, the sub-leading coefficient $\mathrm{p}(n,t)$ has the asymptotic expansion
$$
   \mathrm{p}(n,t) = - n\left( {n + \alpha } \right) - \dfrac{{\lambda n \left( {n + \alpha } \right) }}{t} + \dfrac{{\lambda n \left( {n + \alpha } \right)\left( {2 n + \alpha  - \lambda } \right) }}{{{t^2}}}
   +O(t^{-3}).
$$
\end{theorem}
\begin{proof}
The combination of (\ref{ex1}) and (\ref{Phi2}) gives the desired result.
\end{proof}
\begin{theorem}\label{thm}
As $t  \to  + \infty$, the Hankel determinant $D_n(t)$ has the asymptotic expansion
\be\label{dn3}
\ln D_n(t)=\lambda n \ln  t+\widetilde{C}(n)
 + \dfrac{{\lambda n  \left( {n + \alpha } \right) }}{t} - \dfrac{{\lambda n  \left( {n + \alpha } \right)\left( {2 n + \alpha  - \lambda } \right) }}{{2{t^2}}}+O(t^{-3}),
\ee
where the constant term $\widetilde{C}(n)$ is given by
$$
\widetilde{C}(n)=\ln \dfrac{G(n+1)G(n+\alpha+1)}{G(\alpha+1)}
$$
and  $G(\cdot)$  is the Barnes $G$-function \cite{ew}.
\end{theorem}
\begin{proof}
From (\ref{Phi2}) and in view of the fact that ${H _n}(t) = t\frac{d}{{dt}}\ln {D_n}(t)$, we have
\be\label{Dt2}
\ln D_n(t)=\lambda n \ln  t+\widetilde{C}(n)
 + \dfrac{{\lambda n  \left( {n + \alpha } \right) }}{t} - \dfrac{{\lambda n  \left( {n + \alpha } \right)\left( {2 n + \alpha  - \lambda } \right) }}{{2{t^2}}}+O(t^{-3}),
\ee
where $\widetilde{C}(n)$ is an integration constant, independent of $t$, to be determined. Using (\ref{bd}) we have
\be\label{lnb}
\ln \beta_n(t)=\ln D_{n+1}(t)+\ln D_{n-1}(t)-2\ln D_{n}(t).
\ee
By (\ref{d22}), it follows that
\be\label{beta3}
\ln \beta_n(t)=\ln \left[ {n\left( {n + \alpha } \right)} \right] + \dfrac{{2\lambda }}{t} - \dfrac{{\lambda \left( {  6n + 3\alpha  - \lambda } \right)}}{{{t^2}}} + O(t^{-3}).
\ee
Substituting (\ref{Dt2}) and (\ref{beta3}) into (\ref{lnb}) and comparing the constant terms on both sides, we find
\be\label{CC1}
\widetilde{C}(n+1)+\widetilde{C}(n-1)-2\widetilde{C}(n)=\ln \left[ {n\left( {n + \alpha } \right)} \right].
\ee
Since
$$
D_0(t)=1,\qquad D_1(t)={\mu _0}(t),
$$
we have
$$
\ln D_0(t)=0,\qquad \ln D_1(t)=\ln {\mu _0}(t).
$$
From (\ref{mu0}) we find as $t  \to  + \infty$,
$$
{\mu _0}(t)=t^{\lambda } \left[\Gamma (\alpha +1)+\frac{\lambda  \Gamma (\alpha +2)}{t}+\frac{\lambda(\lambda -1) \Gamma (\alpha +3)}{2 t^2}+\frac{\lambda(\lambda -1)(\lambda -2) \Gamma (\alpha +4)}{6 t^3}+O(t^{-4})\right].
$$
It follows that
\be\label{C10}
\widetilde{C}(0)=0,\qquad \widetilde{C}(1)=\ln \Gamma ( \alpha+1 ).
\ee
The second-order difference equation (\ref{CC1}) with the initial conditions (\ref{C10}) has a unique solution given by
$$
\widetilde{C}(n)=\ln \dfrac{G(n+1)G(n+\alpha+1)}{G(\alpha+1)},
$$
and we establish the theorem.
\end{proof}
\begin{remark}
It is interesting to find that the constant term in Theorem \ref{thm} is the same with the one in the asymptotics of the Hankel determinant for the generalized Airy weight as $t\rightarrow -\infty$ studied in \cite[Theorem 5.3]{F}.
\end{remark}
\begin{theorem}
As $t  \to  + \infty$, the normalized constant $h_n(t)$ has the asymptotic expansion
\begin{align}
\ln h_n(t)=&\:\lambda \ln t +\ln \left( \Gamma ( n+1) \Gamma (n+ \alpha+1 )\right)+ \frac{{\lambda\left( { 2 n+ \alpha+1 } \right) }}{t}\no\\[8pt]
&- \frac{{\lambda \left[6 n^2+2n (3 \alpha - \lambda +3) +(\alpha +1) (\alpha -\lambda +2)\right]}}{{2{t^2}}}+O(t^{-3}).\no
\end{align}
\end{theorem}
\begin{proof}
Substituting (\ref{dn3}) into (\ref{hn}) gives the desired result.
\end{proof}

In the end, we point out that the large $t$ asymptotic expansions of the quantity $H_n(t)= t\frac{d}{{dt}}\ln {D_n}(t)$ and the Hankel determinant $D_n(t)$ obtained in this section can be used to compute the mean, variance and higher order cumulants of the mutual information for single-user MIMO systems; see \cite[Section IV]{cm}. Moreover, we give a rigorous derivation of the large $t$ asymptotic expansion for the quantity $H_n(t)$ in Theorem \ref{th1t} using the mathematical induction and the relation to the recurrence coefficients of corresponding orthogonal polynomials, while Chen and McKay \cite[p. 4610]{cm} derived the special case ($\alpha=0$) by substituting an assumed series expansion into the Jimbo-Miwa-Okamoto $\sigma$-form of the Painlev\'{e} V equation in Theorem \ref{thm2}.

\section*{Acknowledgments}
This work was partially supported by the National Natural Science Foundation of China under grant number 12001212, by the Fundamental Research Funds for the Central Universities under grant number ZQN-902 and by the Scientific Research Funds of Huaqiao University under grant number 17BS402.

\section*{Conflict of Interest}
The authors have no competing interests to declare that are relevant to the content of this article.
\section*{Data Availability Statement}
Data sharing not applicable to this article as no datasets were generated or analysed during the current study.

\end{document}